\documentclass[a4paper,11pt]{scrartcl}\setkomafont{disposition}{\normalfont}
\usepackage[utf8]{inputenc}
\usepackage[T1]{fontenc}
\usepackage[UKenglish]{babel}
\usepackage[sc,osf]{mathpazo}
\usepackage[scaled=0.86]{berasans}
\usepackage[scaled=1.03]{inconsolata}
\usepackage{graphicx}
\usepackage{amsmath,amssymb,amsthm}
\usepackage{xspace}
\usepackage{mathtools}
\usepackage{dsfont}
\usepackage{xcolor}
\usepackage{csquotes}
\usepackage{cite}
\usepackage[colorlinks]{hyperref}

\let\originalleft\left
\let\originalright\right
\renewcommand{\left}{\mathopen{}\mathclose\bgroup\originalleft}
\renewcommand{\right}{\aftergroup\egroup\originalright}

\newcommand{\ket}[1]{| #1 \rangle}

\newcommand{\norm}[1]{\|#1\|}
\newcommand{\abs}[1]{\left|#1\right|}
\newcommand{\de}[1]{\left(#1\right)}

\newcommand{\mathand}{\quad\text{and}\quad}

\newcommand{\be}{\begin{equation}}
\newcommand{\ee}{\end{equation}}
\newcommand{\eg}{e.g.\@\xspace}

\newtheorem{theorem}{Theorem}
\newtheorem{lemma}[theorem]{Lemma}
\mathtoolsset{centercolon}

\hypersetup{pdftitle={{Probability in two deterministic universes}}}

\begin{document}
\title{Probability in two deterministic universes}
\author{Mateus Araújo\\
\textit{\small Institute for Theoretical Physics, University of Cologne, Zülpicher Straße 77,}\\\textit{\small50937 Cologne, Germany}}
\date{\today}
\maketitle
\begin{abstract}
How can probabilities make sense in a deterministic many-worlds theory? We address two facets of this problem: why should rational agents assign subjective probabilities to branching events, and why should branching events happen with relative frequencies matching their objective probabilities. 

To address the first question, we generalise the Deutsch-Wallace theorem to a wide class of many-world theories, and show that the subjective probabilities are given by a norm that depends on the dynamics of the theory: the $2$-norm in the usual Many-Worlds interpretation of quantum mechanics, and the $1$-norm in a classical many-worlds theory known as Kent's universe. To address the second question, we show that if one takes the objective probability of an event to be the proportion of worlds in which this event is realised, then in most worlds the relative frequencies will approximate well the objective probabilities. This suggests that the task of determining the objective probabilities in a many-worlds theory reduces to the task of determining how to assign a measure to the worlds. 
\end{abstract}

We are used to think of a probabilistic situation as one where either an event $E$ or an event $\neg E$ happens, with probabilities $p$ and $1-p$. In many-world theories, however, probabilistic situations are treated as deterministic branching situations, where some worlds are created where event $E$ happens, and some worlds are created where event $\neg E$ happens. Does it still make sense to assign probabilities $p$ and $1-p$ to the worlds with events $E$ and $\neg E$?

This question was already raised at the inception of the Many-Worlds interpretation by Everett in 1957 \cite{everett57}. Several early attempts were made to understand probability from a frequentist point of view \cite{finkelstein63,hartle68,dewitt70,graham73,farhi89} that were mathematically mistaken \cite{squires90,caves04b}. Progress had to wait until 1999, when Deutsch proposed a proof of the Born rule from decision-theoretical assumptions, that was subsequently clarified and improved upon by Wallace \cite{deutsch99,wallace03,wallace07,wallace09}. In Deutsch and Wallace's proofs probabilities are understood as tools rational agents use to make decisions about branching situations, in analogy to single-world decision theory. 

Other successful derivations of the Born rule in the Many-Worlds interpretation have been proposed since then. Żurek proposed a particularly simple one, that however does not address the question of what probabilities should mean in a many-worlds theory \cite{zurek03}. More recently Vaidman, Carroll, and Sebens \cite{vaidman11,carroll14} provided proofs that do offer such an explanation: in their argument probabilities are a measure of the lack of knowledge an agent has about in which world they are, at the instant of time after they performed a measurement but before they learned the outcome. This explanation is not satisfactory, however, as probabilities must also make sense before a measurement is performed, where no uncertainty arises.


In this paper, we investigate the probability problem by focussing on general many-world theories that have a measurement-like branching situation, where a finite number of possible outcomes can occur, and one set of worlds is created with each outcome. Each set of worlds has a coefficient associated to them, with the only constraint that if the coefficient is zero, then no worlds exist with that outcome. 

Two concrete many-world theories that fit into this framework are the usual Many-Worlds interpretation, where the coefficients are the usual complex amplitudes, and a classical many-worlds theory introduced by Adrian Kent in Ref.~\cite{kent09}. There, agents live in a deterministic computer simulation, and the coefficients are non-negative integers that literally indicate the number of copies of the simulation that are run with each outcome after branching. We generalise the Deutsch-Wallace theorem to these many-world theories, and show that it implies, in Kent's universe, that the subjective probabilities must come from the $1$-norm of the coefficients, and in Many-Worlds from the $2$-norm.

The decision-theoretical approach has been criticized by several authors   \cite{gill03,baker06,kent09,albert10,dawid13,mandolesi15,mandolesi18}. Besides objections against specific aspects of Deutsch's and Wallace's proofs that do not apply to the current work, the criticism focuses on three main points: (i) the claim that to derive the branching structure of Many-Worlds via decoherence one needs to assume the Born rule in the first place, which would make the Deutsch-Wallace argument circular, (ii) the observation that there are irrational agents who do not use probabilities to make their decisions, and (iii) the claim that it is incoherent to use decision theory to derive probabilities that are actually objective. With regards to (i), it is not true that the derivation of the branching structure depends on the Born rule. It depends only on being able to say that quantum states that are close in some metric give rise to similar physics, which is standard scientific practice \cite{wallace12}. In any case, in this paper we are not concerned with deriving the branching structure, but rather with making sense of probabilities given that such a branching structure exists. With regards to (ii) and (iii), one should note that these criticisms apply equally well to many-world and single-world versions of decision theory, so they are actually against the whole idea of subjective probability. To a subjectivist they ring hollow, as for them probability is nothing but a tool used by rational agents. Nevertheless, it is hard to deny that there is something objective about quantum probabilities.

To address this point, we show that for any many-worlds theory where a fairly general measure can be assigned to the sets of worlds with a given outcome, an analogue of the law of large number can be proved, which says that in most worlds the relative frequency of some outcome will be close to the proportion of worlds with that outcome. This suggests that we should define the objective probability of an outcome as the proportion of worlds with that outcome. While in Kent's universe it is obvious how to calculate the proportion of worlds, and thus the objective probabilities, this is not the case in Many-Worlds. There, the best this argument can do is say that one should take the relative frequencies as a guess for the proportion of worlds, as in most worlds one will be right.

We also explore a non-symmetric branching scenario where Kent's universe and Many-Worlds assign fundamentally different proportions of worlds to each outcome, and thus the analogy between these many-world theories breaks down. We argue that this is due to the fact that the number of worlds in Kent's universe increases upon branching, while in Many-Worlds the measure of worlds is conserved. To support this argument we introduce a reversed version of Kent's universe where branching conserves the number of worlds, and show that it assign proportions of worlds compatible with Many-Worlds. This breakdown suggests that one of the key rationality principles that Deutsch and Wallace used in their proofs of the Born rule -- called substitutibility by Deutsch and diachronic consistency by Wallace -- is not valid in Kent's universe, and therefore shouldn't be accepted for every many-worlds theory. The proof of the generalised Deutsch-Wallace theorem presented here does not assume this principle, but instead derives it as a theorem in Many-Worlds and the reverse Kent's universe.

\section{Naïve decision theory}

The generalised many-world theories we shall be concerned with have very little structure. We only need them to admit a measurement-like branching situation with $n$ possible outcomes, where after branching $n$ sets of worlds are created, one with each outcome. Each outcome $i$ has a complex coefficient $c_i$ associated to it whose precise meaning depends on the physics underlying the many-worlds theory, but obeys the general constraint that if the coefficient is equal to zero, then no worlds are created with this outcome.

We want to use this branching situation to play a game where an agent receives a reward $r_i$ in the worlds that are created with outcome $i$. The game is then defined by a vector of coefficients $\mathbf{c} \in \mathbb C^n$, and a vector of rewards $\mathbf{r} \in \mathbb R^n$, and can be represented by the $n \times 2 $ matrix
\be
G = (\mathbf{c},\mathbf{r}) = \begin{pmatrix} 
	      c_1 & r_1 \\
	      \vdots & \vdots \\
	      c_n & r_n
            \end{pmatrix},
\ee
where line $i$ specifies the coefficient and reward associated to the worlds with outcome~$i$.

We are concerned with the price at which a rational agent with full knowledge about the situation would accept to buy or sell a ticket for playing a game $G$. This \emph{fair} price is called the value of the game, and is denoted by $V(G)$. From this price we shall infer subjective probabilities that are attributed by the agent to the outcomes. To determine this fair price we shall use a fairly basic decision theory, taken from Ref.~\cite{joyce08}, that is essentially a formalisation of the Dutch book argument.

The first rationality axiom we demand is that if in all worlds the game gives the same reward $r_1$, the agent must assign value $r_1$ to the game. If the value were higher, the agent would lose money in all worlds when buying a ticket for this game; if it were lower, the agent would lose money in all worlds when selling a ticket for this game. This axiom implies that the agent is indifferent to branching per se, and only cares about the rewards their future selves get. This might be false, as we can easily conceive of an agent which is reluctant about branching, and would only accept to pay a value smaller than $r_1$ to play this game. In this case the agent would also be reluctant to sell a ticket for this game, and demand a price higher than $r_1$, so there is no fair price that can be agreed upon, and presumably the game simply won't be played. We deem this behaviour to be irrational. This is specially so in theories where branching is ubiquitous and unavoidable, such as Many-Worlds.

The axiom is then
\begin{itemize}
 \item \textbf{Constancy} \textit{If for some game $G = (\mathbf{c},\mathbf{r})$ all rewards $r_i$ are equal to $r_1$, then $V(G) = r_1$.}
\end{itemize}
The second axiom we demand is that if for a pair of games with equal coefficients $G$ and $G'$ the rewards of the first game are larger than or equal to the rewards of the second game in all worlds, then the agent must value the first game no less than the second game. This could be false, for example, if some agent actively wants to hurt their future selves that receive rewards smaller than the maximal one, and accepts to pay more for a game where this happens. We deem this behaviour to be irrational.

The axiom is then\footnote{\textbf{Dominance} will only be needed to prove Theorem \ref{thm:subjectiveprobability}, that says that rational agents in a many-worlds theory assign subjective probabilities to the worlds. All other results follow without it, including Theorem \ref{thm:1born} and Theorem \ref{thm:2born}, that say that rational agents in Kent's universe and Many-Worlds bet according to the appropriate version of the Born rule.}
\begin{itemize}
 \item \textbf{Dominance} \textit{Let $G = (\mathbf{c},\mathbf{r})$ and $G' = (\mathbf{c},\mathbf{r'})$ be two games that differ only in their rewards. If $\mathbf{r} \ge \mathbf{r'}$, then $V(G) \ge V(G')$.}
\end{itemize}
The third axiom we need is that it shouldn’t matter if a ticket for a game with rewards $\mathbf{r} + \mathbf{r'}$ is sold at once, or broken down into first a ticket for the a game with rewards $\mathbf{r}$ followed by a ticket for another game with rewards $\mathbf{r'}$, where these games are played using the same branching situation. If the price of the composed ticket were higher than the sum of the values of the individual tickets, the agent would lose money in all worlds when buying the composed ticket and selling the individual tickets. On the other hand, if the price of the composed ticket were lower, the agent would lose money in all worlds when buying both individual tickets and selling the composed ticket. The axiom is then
\begin{itemize}
 \item \textbf{Additivity} \textit{Let $G = (\mathbf{c},\mathbf{r})$ and $G' = (\mathbf{c},\mathbf{r'})$ be two games that differ only in their rewards. Then the game $G'' = (\mathbf{c},\mathbf{r}+\mathbf{r'})$ has value $V(G'') = V(G) + V(G').$}
\end{itemize}
This \textbf{Additivity} axiom is eminently reasonable in the scenario considered, where the agent should be willing to act as bookie and bettor in the game. In decision theory, however, one usually considers scenarios where an agent is simply offered the gamble and is trying to decide how much to pay for it, without needing to act as a bookie. In this case it is not irrational to set $V(\mathbf{c},\mathbf{r}+\mathbf{r'}) < V(\mathbf{c},\mathbf{r})+V(\mathbf{c},\mathbf{r'})$; it is in fact necessary to avoid pathological decisions such as Pascal's Wager and the St. Petersburg paradox. As shown by Wallace, this scenario can be dealt with using a more sophisticated decision theory, such as Savage's \cite{wallace09,savage54}. But we believe that such decision-theoretical sophistication distracts from the physical arguments at hand, and will stick with the simpler bettor/bookie scenario. One should note, anyway, that even in Savage's theory \textbf{Additivity} is a good approximation when the values being gambled are small compared to the bettor's total wealth, and one could constrain the analysis to games in which this condition is satisfied.

The fourth and last axiom we need from classical decision theory is that the value function must be continuous on the coefficients and rewards:
\begin{itemize}
 \item \textbf{Continuity} \textit{Let $(G_k)$ be a sequence\footnote{For concreteness, convergence is defined in the metric induced by the norm $\norm{G} := \norm{\mathbf{c}}_3 + \norm{\mathbf{r}}_1$.} of games such that $\lim_{k\to\infty}G_k = G$. Then $V(G) = \lim_{k\to\infty} V(G_k)$}
\end{itemize}
It is here just for mathematical convenience. It can be left out entirely if one is happy to restrict the rewards to be rational numbers (as more realistically they would be integer multiples of eurocents) and the coefficients to belong to some countable subset of the complex numbers that depends on the particular many-worlds theory (as it would be physically suspect to demand them to be specified with infinite precision).

The axioms presented up to this point are already enough to imply that the agent must assign subjective probabilities to the outcomes of the measurement in the game. The proof is an elementary exercise in decision theory, but we shall include it here anyway because it is short, enlightening, and mostly unfamiliar to physicists:
\begin{theorem}\label{thm:subjectiveprobability}
Let $G = (\mathbf{c},\mathbf{r})$ be a game with $n$ outcomes. A rational agent must assign it value
\be V(G) = \sum_{i=1}^n V(\mathbf{c},\mathbf{e}_i)r_i,\ee
where the vectors $\mathbf{e}_i$ are the standard basis, and $V(\mathbf{c},\mathbf{e}_i)$ are the subjective probabilities, as
\be V(\mathbf{c},\mathbf{e}_i) \ge 0 \mathand \sum_{i=1}^n V(\mathbf{c},\mathbf{e}_i) = 1\ee
\end{theorem}
\begin{proof}
By \textbf{Additivity}
\be V(\mathbf{c},\mathbf{r}) = \sum_{i=1}^n V(\mathbf{c},r_i\mathbf{e}_i),\ee
so we only need to compute the value of games with a single non-zero reward, e.g. $(\mathbf{c},r_1\mathbf{e}_1)$, upon which we shall now focus. Using \textbf{Additivity} again, we see that for any positive integer $b$ we have that
\be V(\mathbf{c},b r_1\mathbf{e}_1) = bV(\mathbf{c},r_1\mathbf{e}_1), \ee
and setting $r_1=1/b$ shows us that
\be V\de{\mathbf{c}, \frac1b \mathbf{e}_1} = \frac1bV(\mathbf{c},\mathbf{e}_1). \ee
Another application of \textbf{Additivity} lets us conclude that for any positive integer $a$ we have that
\be V\de{\mathbf{c}, \frac ab \mathbf{e}_1} = \frac abV(\mathbf{c},\mathbf{e}_1), \ee
so this homogeneity is valid for any positive rational number. To extend it for any rational number $q$ we use \textbf{Constancy} to conclude that $V(\mathbf{c},\mathbf{0}) = 0$, where $\mathbf{0}$ is the vector of all zeroes, and \textbf{Additivity} again to show that 
\be
0 = V(\mathbf{c},\mathbf{0}) = V\de{\mathbf{c}, q \mathbf{e}_1} + V(\mathbf{c},-q\mathbf{e}_1),
\ee
and thefore that
\be
V(\mathbf{c},-q\mathbf{e}_1) = - V\de{\mathbf{c}, q \mathbf{e}_1}.
\ee
To extend this to any real number $\mu$, let $(q_k)$ be a sequence of rational numbers converging to it. Then by \textbf{Continuity}
\be V(\mathbf{c},\mu\mathbf{e}_1) =  \lim_{k\to\infty} V(\mathbf{c},q_k \mathbf{e}_1) = \lim_{k\to\infty} q_k V(\mathbf{c}, \mathbf{e}_1) = \mu V(\mathbf{c},\mathbf{e}_1),\ee
and therefore for any reward vector $r$ we have that 
\be V(\mathbf{c},\mathbf{r}) = \sum_{i=1}^n r_i V(\mathbf{c},\mathbf{e}_i),\ee
so the computation of the value of an arbitrary game reduces to the computation of the values of the elementary games $(\mathbf{c},\mathbf{e}_i)$, which are by definition the subjective probabilities. To see that they are positive and normalised first notice that by \textbf{Dominance}
\be V(\mathbf{c},\mathbf{e}_i) \ge V(\mathbf{c},\mathbf{0}) = 0,\ee
and that by \textbf{Constancy}
\be V(\mathbf{c},\mathbf{1}) = 1 = \sum_{i=1}^n V(\mathbf{c},\mathbf{e}_i),\ee
where $\mathbf{1}$ is the vector of all ones.
\end{proof}

This is as far as we can go by using the naïve decision theory. Rational agents must reason about the games by assigning a probability to each outcome, but the decision theory is silent about what the probabilities must be. 

One can go further, though, by assuming an additional axiom reminiscent of the disreputable\footnote{It is untenable in classical probability because in a given situation there are often several different plausible symmetries that give rise to different probability assignments. This problem does not arise here, as the symmetry at hand will be the one between the coefficients of the worlds, or the amplitudes of the quantum state.} Principle of Indifference from classical probability \cite{vanfraassen89}. It states that one should assign uniform probabilities to symmetric situations: coins and dice should be regarded as equiprobable because there is no preferred side of the coin or face of the die. 

In a deterministic single-world theory the principle is nonsense, as a truly symmetric situation would be incapable of producing an output. Deterministic single world theories that are capable of producing outcomes, such as Bohmian mechanics, must have a variable breaking the symmetry. But a rational agent that knows its value should definitely not assign uniform probabilities, but rather probability one to the outcome that will actually happen. 

Probabilistic single-world theories fare a bit better, as one can have a perfectly symmetric situation before the measurement. But after the measurement the symmetry is necessarily broken, as only one outcome actually happens, and a rational agent that knows which should definitely not assign uniform probabilities.

In both deterministic and probabilistic single-world theories, therefore, an indifference principle can only apply to agents with restricted knowledge: about the present, in the deterministic case, or about the future, in the probabilistic case. Only in many-world theories can the symmetry remain rigorously unbroken, even after the measurement is made, and agents with full knowledge about the game can apply the indifference principle. Note that even if the agent's future selves could relay information to the past about which outcome they experience it wouldn't help, as the agent already knows that they will have future selves experiencing all outcomes.

The axiom is then
\begin{itemize}
 \item \textbf{Indifference} \textit{Let $G = (\mathbf{c},\mathbf{r})$ and $\sigma$ be a permutation. Then the game with permuted coefficients and rewards $G_\sigma = (\sigma(\mathbf{c}),\sigma(\mathbf{r}))$ has value $V(G_\sigma) = V(G)$}
\end{itemize}
Notice that this axiom is not merely saying that it does not matter whether we write down $1$ or $2$ to label the outcome of a coin toss, but rather that it does not matter if we exchange the coefficients and the rewards of the set of worlds with outcome $1$ with those of the set of worlds with outcome $2$. Worlds $1$ and worlds $2$ are physically different, if nothing else for having different measurement results. \textbf{Indifference} implies that whatever these differences are, the coefficients and rewards are the only relevant ones. This is clearly not true if outcome $1$, and not $2$, happens.

This axiom suffices to prove that symmetric games have uniform probabilities:
\begin{lemma}[Symmetry]\label{thm:symmetry}
Let $G = (\mathbf{c},\mathbf{r})$ be a game with $n$ outcomes such that all coefficients $c_i$ are equal to $c_1$. Then \be V(G) = \frac1n \sum_{i=1}^n r_i.\ee
\end{lemma}
\begin{proof}
Let
\be G_{\sigma_i} = (\sigma_i(\mathbf{c}),\sigma_i(\mathbf{r})) = (\mathbf{c},\sigma_i(\mathbf{r}))\ee
be a version of $G$ with the $i$th cyclic permutation applied to the coefficients and rewards, so that by \textbf{Indifference} $V(G_{\sigma_i}) = V(G)$. If one then defines the reward vector
\be \boldsymbol{\rho} = \sum_{i=1}^n \sigma_i(\mathbf{r}),\ee
all its components are equal to $\sum_{i=1}^n r_i$, so by \textbf{Constancy} the game $\Gamma = (\mathbf{c},\boldsymbol{\rho})$ has value
\be V(\Gamma) = \sum_{i=1}^n r_i,\ee
but also by \textbf{Additivity}
\be V(\Gamma) = \sum_{i=1}^n V(\mathbf{c},\sigma_i(\mathbf{r})) = \sum_{i=1}^n V(G_{\sigma_i}) = nV(G),\ee
and therefore 
\be V(G) = \frac1n \sum_{i=1}^n r_i.\ee
\end{proof}

Now, to deal with games with unequal coefficients, we shall do a fine-graining argument to reduce them to symmetric games. How the fine-graining works depends on the precise physics of the many-worlds theory, so we'll have different fine-graining axioms for different theories.

\section{Kent's universe}\label{sec:kentsuniverse}

In Kent's universe the agent lives in a deterministic computer simulation run by some advanced civilization. The branching happens at the press of a button that is displayed on a wall, that in addition contains a list of non-negative integers $\mathbf{m} = (m_1,\ldots,m_n)$ and real numbers $\mathbf{r} = (r_1,\ldots,r_n)$. These integers play the role of the coefficients in the game, and $m_i$ is literally the number of successor worlds which are created with reward $r_i$ in them.

In this case it is clear how to do fine-graining: if one plays the game 
\be G = \begin{pmatrix} 1 & r_1 \\ 2 & r_2 \end{pmatrix} \ee
or the game
\be G' = \begin{pmatrix} 1 & r_1 \\ 1 & r_2 \\ 1 & r_2 \end{pmatrix}, \ee
in both games three successor worlds are created, one with reward $r_1$ in it, and two with reward $r_2$. The only difference is that while in game $G$ both worlds with reward $r_2$ are labelled with outcome $2$, in game $G'$ one of the worlds with reward $r_2$ is labelled with outcome $2$ and the other with outcome $3$. We postulate that this difference does not matter, and that $V(G') = V(G)$. Since $G'$ is symmetric, we can evaluate $V(G')$ using Lemma \ref{thm:symmetry}, and thus determine that
\be V(G) = \frac13 r_1 + \frac 23 r_2.\ee
Generalising this argument, we can determine the value of \emph{any} game by fine-graining it into a symmetric game. The formal postulate we need is
\begin{itemize}
 \item \textbf{1-Fine-graining} \textit{Let 
 \be
 G = \begin{pmatrix}
 m_1 & r_1 \\
 m_2 & r_2 \\
 \vdots & \vdots \\
 m_n & r_n 
 \end{pmatrix}
 \ee
 be a game with $n$ outcomes. Then for any non-negative integers $m_{11},m_{12}$ such that $m_{11} + m_{12} = m_1$ the fine-grained game with $n+1$ outcomes 
 \be
 G' = \begin{pmatrix}
 m_{11} & r_1 \\
 m_{12} & r_1 \\
 m_2 & r_2 \\
 \vdots & \vdots \\
 m_n & r_n
\end{pmatrix}
\ee
has value $V(G') = V(G)$}
\end{itemize}
\begin{theorem}\label{thm:1born}
Let $G = (\mathbf{m},\mathbf{r})$ be a game. A rational agent in Kent's universe must assign it value
\be V(G) = \frac1{\norm{\mathbf{m}}_1}\sum_{i=1}^n m_i r_i.\ee
\end{theorem}
\begin{proof}
Starting from an arbitrary $n$-outcome game $G = (\mathbf{m},\mathbf{r})$, we can repeatedly apply \textbf{1-Fine-graining} to take it to a symmetric game $G' = (\mathbf{1}^{(\sum_{i=1}^nm_i)},\mathbf{r'})$, where $\mathbf{1}^{(\sum_{i=1}^nm_i)}$ is the vector of $\sum_{i=1}^nm_i$ ones. The reward vector $\mathbf{r'}$ consists of $m_1$ rewards $r_1$, $m_2$ rewards $r_2$, etc., and can be written as $\mathbf{r'} = \bigoplus_{i=1}^n \mathbf{r}_i^{(m_i)}$. Its value $V(G')$ can then by calculated via Lemma \ref{thm:symmetry}, and is given by
\be V(G') = \frac1{\sum_{i=1}^nm_i}\sum_{i=1}^n \sum_{j=1}^{m_i} r_i = \frac1{\norm{\mathbf{m}}_1}\sum_{i=1}^n m_i r_i.\ee
\end{proof}

\section{Many-Worlds}\label{sec:manyworlds}

In the Many-worlds interpretation a measurement is a unitary transformation, that via decoherence gives rise to superpositions of quasi-classical worlds. We shall not discuss how this happens, as this has been extensively explored elsewhere \cite{saunders93,wallace01,zurek03b,hartle08}. Rather, we shall take for granted that this is indeed the case, and that one can instantiate the games by doing measurements on quantum states, with the complex amplitudes playing the role of the coefficients.

The game is started by giving the agent a state
\be \ket{\psi} = \sum_{i=1}^n \alpha_i \ket{i},\ee
where the $\ket{i}$ are distinguishable states of some infinite-dimensional degree of freedom, such as position. The agent then does a measurement in the computational basis: a unitary transformation that acts on the basis state $\ket{i}$ together with the measurement device in the ready state $\ket{M_?}$ and takes the measurement device to the state with outcome $i$ called $\ket{M_i}$:
\be \ket{i}\ket{M_?} \mapsto \ket{i}\ket{M_i}.\ee
The game is concluded by giving the agent reward $r_i$ in the worlds with outcome $i$, with the whole process taking the initial state $\sum_{i=1}^n \alpha_i \ket{i}\ket{M_?}\ket{r_?}$ to the final state 
\be \ket{G} = \sum_{i=1}^n \alpha_i \ket{i}\ket{M_i}\ket{r_i}. \ee
Since the measurement is done in a fixed basis, the game is completely defined by the vector of coefficients $\boldsymbol{\alpha}$ and the vector of rewards $\mathbf{r}$, and so we can represent it by the matrix $G = (\boldsymbol{\alpha},\mathbf{r})$.

We want to proceed as in Kent's universe, and use a fine-graining postulate to reduce arbitrary games to symmetric games. It is not so obvious, however, how to do fine-graining in the Many-Worlds picture, as we cannot count worlds in order to say that games are equivalent if they assign the same number of worlds to each reward. As argued by Wallace, the continuous nature of quantum mechanics and the arbitrariness of the border between worlds make it impossible to count them in a physically meaningful way \cite{wallace07}. One could try, instead, to arbitrarily postulate that there is one world for each measurement outcome with nonzero coefficient, as is often proposed\footnote{To the best of our knowledge there has been no attempt to calculate the number of worlds in an even remotely realistic model of a measurement.} \cite{graham73,price10}.

To see that this is untenable, consider the historical Stern-Gerlach experiment \cite{gerlach22}, that measured the spin\footnote{Note that Stern and Gerlach were not aware that they were measuring spin, rather they interpreted the experiment as a proof of Bohr-Sommerfelds spatial quantization hypothesis.} of silver atoms by letting them accumulate on a glass plate; if some atom ended up on the left hand side of the glass plate it had spin $+\hbar\frac12$, and if ended up on the right hand side it had spin $-\hbar\frac12$. Their apparatus, however, was not doing a left/right measurement, but rather a much more precise position measurement, that needs to be coarse-grained in order to obtain a two-outcome spin measurement. Should we then consider the large number of distinguishable positions to be the number of worlds produced in this experiment? But if so, what would we make of a more modern Stern-Gerlach experiment, that instead of the glass plate uses a Langmuir-Taylor detector, a single hot wire that is scanned across the neutral atom beam \cite{phipps32,frisch33}? Should it be taken to produce two worlds, one in which the atom hit the wire and another in which it didn't? And what if we do a Frankenstein version of the Stern-Gerlach experiment, with a glass plate on the left-hand side and a Langmuir-Taylor detector on the right-hand side? Does it create more worlds with spin $+\hbar\frac12$ than with spin $-\hbar\frac12$?

This should make it clear that the number of outcomes in a measurement is largely arbitrary, and of little relevance. What actually matters is that atoms with spin $+\hbar\frac12$ go left and atoms with spin $-\hbar\frac12$ go right; these different experimental setups are equivalent ways to do the measurement, and we shall base the fine-graining argument in Many-Worlds precisely on this equivalence.

\begin{figure}[ht]
\centering
\includegraphics{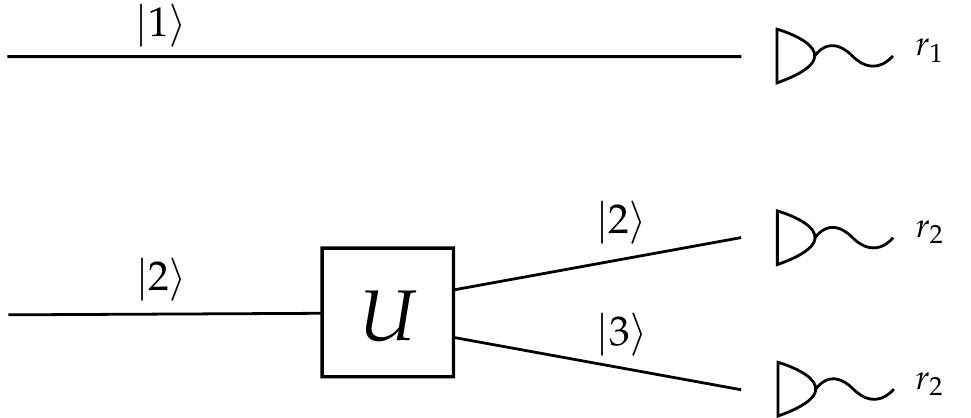}
\caption{One can measure a state in the $\{\ket{1},\ket{2}\}$ basis by applying a unitary $U$ to $\ket{2}$ that takes it to a superposition of $\ket{2}$ and $\ket{3}$, measuring the state in the $\{\ket{1},\ket{2},\ket{3}\}$ basis, and label as $2$ both results $2$ and $3$.}\label{fig:diffraktion}
\end{figure}

Consider then a game where a measurement is made on the state 
\be \ket{\psi} = \alpha\ket{1} + \beta\ket{2},\ee
reward $r_1$ is given in the worlds with outcome $1$, and reward $r_2$ in the worlds with outcome $2$, taking it to the final state
\be \ket{G} = \alpha \ket{1}\ket{M_1}\ket{r_1} + \beta\ket{2}\ket{M_2}\ket{r_2},\ee
which can also be represented as the game matrix
\be G = \begin{pmatrix} \alpha & r_1 \\ \beta & r_2 \end{pmatrix}.\ee
An equivalent way to play this game, shown in Fig. \ref{fig:diffraktion}, is to make a measurement on the state
\be \ket{\psi'} = \alpha\ket{1} + \beta U\ket{2} = \alpha \ket{1} + \beta\gamma\ket{2} + \beta\delta\ket{3}\ee
instead, where an unitary $U$ was applied to $\ket{2}$, and apply the label $2$ and give reward $r_2$ to both outcomes $2$ and $3$, leading to the final state
\be \ket{G'} = \alpha \ket{1}\ket{M_1}\ket{r_1} + \beta\gamma\ket{2}\ket{M_2}\ket{r_2} + \beta\delta\ket{3}\ket{M_2'}\ket{r_2},\ee
where $\ket{M_2'}$ is a measurement result physically distinct from $\ket{M_2}$, but with the same label.

This second way to do the measurement still gives outcomes $1$ and $2$ whenever the initial state is $\ket{1}$ and $\ket{2}$, so it respects the definition of measurement given above and normal practice in real laboratories.

If one does not, however, coarse-grain the outcomes $2$ and $3$ together, but rather leaves them distinct, then we are playing instead the three-outcome game
\be G' = \begin{pmatrix} \alpha & r_1 \\ \beta\gamma & r_2 \\ \beta\delta & r_2 \end{pmatrix}. \ee
We postulate that this labelling choice does not make a difference, so a rational agent must assign the same value to games $G$ and $G'$. This suffices to fine-grain arbitrary games into symmetric games. As an illustration, consider the game 
\be G = \begin{pmatrix} 1 & r_1 \\ 2 e^{i\theta} & r_2 \end{pmatrix}. \ee
Doing the fine-graining with a unitary $U$ that takes the state $\ket{2}$ to
\be U\ket{2} = \frac{e^{-i\theta}}{\sqrt2}\de{\ket{2} + \ket{3}},\ee
we take $G$ to
\be G' = \begin{pmatrix} 1 & r_1 \\ \sqrt{2} & r_2 \\ \sqrt{2} & r_2\end{pmatrix}. \ee
Applying similar unitaries to $\ket{2}$ and $\ket{3}$, we take $G'$ to 
\be G'' = \begin{pmatrix} 1 & r_1 \\ 1 & r_2 \\ 1 & r_2 \\ 1 & r_2 \\ 1 & r_2\end{pmatrix}, \ee
which is a symmetric game, and so $V(G'') = V(G)$ can be evaluated via Lemma \ref{thm:symmetry}, resulting in
\be V(G) = \frac15 r_1 + \frac45 r_2.\ee
We can now prove the Born rule in the general case by formalising this argument via the following axiom:
\begin{itemize}
 \item \textbf{2-Fine-graining} \textit{Let \be G = \begin{pmatrix}
 \alpha_1 & r_1 \\
 \alpha_2 & r_2 \\
 \vdots & \vdots \\
 \alpha_n & r_n 
 \end{pmatrix}\ee be a game with $n$ outcomes. Then for any complex numbers $\alpha_{11},\alpha_{12}$ such that\\ ${\sqrt{|\alpha_{11}|^2 + \abs{\alpha_{12}}^2} = \abs{\alpha_1}}$ the fine-grained game with $n+1$ outcomes \be G' = \begin{pmatrix}
 \alpha_{11} & r_1 \\
 \alpha_{12} & r_1 \\
 \alpha_2 & r_2 \\
 \vdots & \vdots \\
 \alpha_n & r_n
\end{pmatrix} \ee has value $V(G') = V(G)$}
\end{itemize}

\begin{theorem}[Deutsch-Wallace]\label{thm:2born}
Let $G = (\boldsymbol{\alpha},\mathbf{r})$ be a game. A rational agent in Many-Worlds must assign it value
\be V(G) = \frac{1}{\norm{\boldsymbol{\alpha}}^2_2}\sum_{i=1}^n |\alpha_i|^2 r_i\ee 
\end{theorem}
\begin{proof}
Let then $G = (\boldsymbol{\alpha},\mathbf{r})$ be a $n$-outcome game such that the coefficients $\alpha_j$ are of the form $\alpha_j = \sqrt{\frac{a_j}{b_j}}e^{i\theta_j}$, and do a trivial \textbf{2-Fine-graining} in each outcome to take it to a $n$-outcome game $G' = (\boldsymbol{\alpha'},\mathbf{r})$ with coefficients $\alpha_i' = |\alpha_i| = \sqrt{\frac{a_i}{b_i}}$. 

Let then $d=\prod_{i=1}^n b_i$, and define the integer $a_i' = da_i/b_i$ so that $\alpha_i = \sqrt{\frac{a'_i}{d}}$. One can then fine-grain each coefficient $\alpha_i$ into $a'_i$ coefficients equal to $1/\sqrt{d}$, obtaining a symmetric game $G'' = (\mathbf{1}^{(\sum_{i=1}^n a_i')}/\sqrt{d},\mathbf{r'})$ with reward vector $\mathbf{r'} = \bigoplus_{i=1}^n \mathbf{r}_i^{(a_i')}$. By Lemma \ref{thm:symmetry} the value of $G''$ is given by
\be V(G'') = \frac{1}{\sum_{i=1}^n a_i'} \sum_{i=1}^n \sum_{j=1}^{a_i'}r_i = \frac{1}{\sum_{i=1}^n \frac{a_i}{b_i}} \sum_{i=1}^n \frac{a_i}{b_i}r_i = \frac{1}{\norm{\boldsymbol{\alpha}}^2_2}\sum_{i=1}^n |\alpha_i|^2 r_i,\ee
and by \textbf{2-Fine-graining} this is equal to $V(G)$.

To generalise this argument for a game $G = (\boldsymbol{\alpha},\mathbf{r})$ with arbitrary complex coefficients $\alpha_j$, it is enough to notice that numbers of the form $\sqrt{\frac{a}{b}}e^{i\theta}$ are dense in the complex plane. Let then $\boldsymbol{\alpha}_k$ be a sequence of coefficient vectors such that $\lim_{k\to\infty}\boldsymbol{\alpha}_k = \boldsymbol{\alpha}$ and that for all $k$ we have $\alpha_{kl} = \sqrt{\frac{a_{kl}}{b_{kl}}}e^{i\theta_l}$. Let then $G_k = (\boldsymbol{\alpha}_k,\mathbf{r})$. By \textbf{Continuity}, we have that 
\be V(G) = \lim_{k\to\infty} V(G_k) = \lim_{k\to\infty} \frac{1}{\norm{\boldsymbol{\alpha}_k}^2_2}\sum_{l=1}^n \abs{\alpha_{kl}}^2 r_l =  \frac{1}{\norm{\boldsymbol{\alpha}}^2_2}\sum_{j=1}^n |\alpha_i|^2 r_j.\ee
\end{proof}

Note that we did not have to assume that the initial state was normalised, as the proof implies that a factor of $\norm{\boldsymbol{\alpha}}^2_2$ must appear on the denominator of the value function. We did have to assume that the transformations between quantum states are done via unitaries, but this comes from the background assumption of Many-Worlds quantum mechanics. 

In some textbooks one postulates the Born rule as a fundamental principle, and motivates unitary evolution from conservation of probabilities, see \eg Refs.~\cite{sakurai93,feynman77,landau65}. Such an approach would make deriving the Born rule from the unitary evolution a bit circular. Other textbooks, however, postulate unitarity as fundamental, motivated by analogy with Hamiltonian classical mechanics, and consider the Born rule a separate postulate, see \eg Refs.~\cite{vonneumann32,chuang00,peres06}. This latter approach parallels the historical development of the Schrödinger equation and the Born rule \cite{schroedinger26,born26}.

\subsection{Generalisation}\label{sec:generalisation}

To emphasize that the Born rule comes from fine-graining through unitary transformations that preserve the $2$-norm, we'd like to generalise the fine-graining argument to transformations that preserve some other norm. Which norms should we consider? As we show in Appendix \ref{sec:consistency}, some weak consistency conditions imply that fine-graining can only be done if the transformations preserve the $p$-norm of the vectors.

Consider, then, some hypothetical\footnote{One should take seriously the ``hypothetical'' here, as theories where the $p$-norm is preserved are rather pathological. As shown in Ref.~\cite{aaronson04}, the only linear transformations that preserve the $p$-norm of all vectors for $p\neq1,2$ are permutations composed with phases. Here we get around this by only asking the transformation $T$ to preserve the norm of the computational basis.} many-worlds theory where the $p$-norm is preserved. In such a theory we can fine-grain the game
\be G = \begin{pmatrix} 1 & r_1 \\ 2 e^{i\theta} & r_2 \end{pmatrix} \ee
by using a transformation $T$ that takes the state $\ket{2}$ into\footnote{In this example $p$ must be of the form $\log_2 n$ for integer $n\ge2$, but in general any real $p\ge1$ works.}
\be T\ket{2} = \frac{e^{-i\theta}}{2}\sum_{i=1}^{2^p} \ket{i+1}.\ee
The fine-grained game is then
\be G' = \begin{pmatrix} 1 & r_1 \\ \mathbf{1}^{(2^p)} & r_2 \end{pmatrix}, \ee
where $\mathbf{1}^{(2^p)}$ is the vector of $2^p$ ones, and from the analogous argument we conclude that
\be V(G) = \frac1{2^p} r_1 + \frac{2^p-1}{2^p} r_2.\ee
The general $p$-Born rule can then be proven via the following axiom:
\begin{itemize}
 \item \textbf{$p$-Fine-graining} \textit{Let \be G = \begin{pmatrix}
 \alpha_1 & r_1 \\
 \alpha_2 & r_2 \\
 \vdots & \vdots \\
 \alpha_n & r_n 
 \end{pmatrix}\ee be a game with $n$ outcomes. Then for any complex numbers $\alpha_{11},\alpha_{12}$ such that\\ ${\de{|\alpha_{11}|^p + \abs{\alpha_{12}}^p}^{\frac1p} = \abs{\alpha_1}}$ the fine-grained game with $n+1$ outcomes \be G' = \begin{pmatrix}
 \alpha_{11} & r_1 \\
 \alpha_{12} & r_1 \\
 \alpha_2 & r_2 \\
 \vdots & \vdots \\
 \alpha_n & r_n
\end{pmatrix} \ee has value $V(G') = V(G)$}
\end{itemize}
\begin{theorem}
Let $G = (\boldsymbol{\alpha},\mathbf{r})$ be a game. A rational agent in a many-worlds theory where transformations preserve the $p$-norm must assign it value
\be\label{eq:pborn}
V(G) = \frac{1}{\norm{\boldsymbol{\alpha}}^p_p}\sum_{i=1}^n |\alpha_i|^p r_i
\ee
\end{theorem}
\begin{proof}
The proof is similar to the Many-Worlds case, so we shall only sketch it: consider a game $G = (\boldsymbol{\alpha},\mathbf{r})$ with coefficients $\alpha_j = \de{\frac{a_j}{b_j}}^{\frac1p}e^{i\theta_j}$, do a trivial \textbf{$p$-Fine-graining} to get rid of the phase, rewrite the coefficients as 
$\de{\frac{a_i'}{d}}^{\frac1p}$ for $d=\prod_{i=1}^n b_i$ and $a_i' = da_i/b_i$, fine-grain it into a game with $a_i'$ coefficients equal to $\de{1/d}^{\frac1p}$ for each outcome $i$, use Lemma \ref{thm:symmetry} to conclude that the value of the fine-grained game is given by equation \eqref{eq:pborn}, and use \textbf{Continuity} to show that this formula is valid for any game.
\end{proof}

Notice that not only the proof doesn't work for the case of the $\max$-norm (often referred to as $p=\infty$), but the result is also false, as it is not possible to fine-grain all games into symmetric games via transformations that preserve the $\max$-norm. For example, any fine-graining of the game 
\be G = \begin{pmatrix} 1 & r_1 \\ 2 & r_2 \end{pmatrix} \ee
will have at least one coefficient $1$ associated to reward $r_1$ and at least one coefficient $2$ associated to reward $r_2$. One could try, nevertheless, to arbitrarily define
\be V(G) = \lim_{p\to\infty}\frac{1}{\norm{\boldsymbol{\alpha}}^p_p}\sum_{i=1}^n |\alpha_i|^p r_i = \frac{1}{\#M}\sum_{i\in M}r_i,\ee
where $M$ is the set of $i$ such that $|\alpha_i| = \max_j |\alpha_j|$. This $\max$-Born rule would be consistent with \textbf{Constancy}, \textbf{Dominance}, \textbf{Additivity}, and \textbf{Indifference}, but not with \textbf{Continuity}.

\section{Objective probabilities}\label{sec:objectiveprobabilities}

\begin{flushright}
\small \hspace{0.30\textwidth} \textit{We do not accept that the behaviour of a rational decision maker should play a role in modelling physical systems -- Richard Gill \cite{gill03}}.
\end{flushright}

Up to now the discussion has been exclusively about subjective probabilities. We have argued only that it would be irrational to assign probabilities different than those given by the generalised Deutsch-Wallace theorem because one would unjustifiably break the symmetries of the theory. We haven't argued, though, that it would be irrational to assign different probabilities because they would predict different relative frequencies. Such an argument would be suspect: why should these explicitly subjective probabilities be connected to the clearly objective relative frequencies? Shouldn't relative frequencies be connected to objective probabilities instead?

The situation is not as bad as it sounds, because although subjective, these probabilities depend only on the coefficients of the game, which are objective, and their derivation is done from defensible rationality axioms together with facts about the world and the theory describing it; it would be rather surprising if these subjective probabilities turned out to be completely fantastic. In fact, such well-grounded subjective probabilities are even taken to be equal to the objective probabilities by Lewis' Principal Principle \cite{lewis80}, which states that 
\be \text{Pr}_o(E) = \text{Pr}_s(E|HT),\ee
that is, the objective probability of an event $E$ is equal to the subjective probabilities assigned to $E$ by a rational agent that knows the history of the world up to this point $H$ and the theory that describes the world $T$.

This is still a bit unsatisfactory, because the relationship between objective probabilities and relative frequencies shouldn't depend on the opinions of rational agents or even their existence. We believe, after all, that radioactive elements have been decaying on Earth much before humans appeared to reason about them, and it would be absurd to assume that the relative frequencies budged at all when we arrived.

We can do better, though. We shall propose a definition of objective probability in terms of the proportion of worlds in which an event is realised, and show that the connection of these objective probabilities with relative frequencies is mathematically identical to the one in single-world theories. In this way they fulfil the main role\footnote{We talk about the roles objective probabilities play instead of the definition of objective probabilities because this is the best we can do. Unlike subjective probabilities, there is no widely accepted definition of objective probability that we could try to satisfy.} objective probabilities must play, as defended by Saunders in Ref.~\cite{saunders10}.

This connection, in single-world theories, is given by the law of large numbers, that says roughly that for a large number of trials the relative frequencies will be close to the objective probability, with high objective probability\footnote{One often hears a different story, that in the limit of an infinite number of trials the relative frequency is \emph{equal} to the objective probability, full stop. This is simply mistaken, but given that the mathematics we present here are identical, those that want to insist on this mistake can do it equally well in many-world theories.}. More precisely, if the objective probability of observing some event is $p$, then after $N$ trials the objective probability that the frequency $f_N$ will be farther than $\varepsilon$ away from $p$ is upperbounded by a function that is decreasing in $N$ and $\varepsilon$, that is, that
\be \text{Pr}_o(|f_N-p|>\varepsilon) \le 2e^{-2N\varepsilon^2}\ee
where the precise form of the upper bound is not relevant, but for concreteness we used the one given by the Hoeffding inequality.

To see how this works in many-world theories, we shall first consider how the relative frequencies behave in Kent's universe, and then generalise. Consider then that one is collecting relative frequencies by performing $N$ repetitions of a $n$-outcome measurement, where $m_i$ worlds are created with outcome $i$. After $N$ trials there are $n^N$ sets of worlds, each set identified by a sequence of outcomes $\mathbf{s} = (s_1,\ldots,s_N)$, with $s_i \in \{1,\ldots,n\}$. The number of worlds where a sequence $\mathbf{s}$ is obtained is clearly given by
\be
\#_\mathbf{s} = \prod_{i=1}^N \#_{s_i} = \prod_{i=1}^n m_i^{k_i},
\ee
where $k_i$ is the number of outcomes $i$ in $\mathbf{s}$. If we define then the proportion of worlds with outcome sequence $\mathbf{s}$ as
\be 
\sharp_\mathbf{s} := \frac{\#_\mathbf{s}}{\sum_\mathbf{s'} \#_\mathbf{s'}},
\ee
then the proportion of worlds where the relative frequency $f_N$ of outcome $1$ is $k/N$ is
\be \sharp(k,N) = \binom{N}{k} \sharp_1^k (1-\sharp_1)^{N-k}, \ee
which is formally identical to the binomial distribution, and therefore allows us to prove a law of large numbers for the proportion of worlds, which says that in most worlds the relative frequency of outcome $1$ will be close to $\sharp_1$, or more precisely that
\be
\sharp\Big( |f_N - \sharp_1| \ge \varepsilon \Big) \le 2e^{-2N\varepsilon^2}.
\ee

This suggests that we could try identifying objective probabilities with the proportion of worlds in a general many-worlds theory. Consider then that in such a theory one is collecting frequencies by performing $N$ repetitions of a $n$-outcome measurement as in Kent's universe. Now the branching coefficients are not the number of worlds created with each outcome, however, and as argued in section \ref{sec:manyworlds} we do not think that it is possible to assign a sensible world count. We can, however, assign a sensible world \emph{measure}\footnote{Which does include the counting measure as a particular case, so we are by no means assuming that worlds cannot be counted.}: a function $\Lambda$ that assigns a real number $\Lambda_\mathbf{s} \ge 0$ to the set of worlds with outcome sequence $\mathbf{s}$, and defines the measure of a set of worlds with outcome sequences in some set $S$ to be
\be
\Lambda(S) = \sum_{\mathbf{s} \in S} \Lambda_\mathbf{s}.
\ee 
Since the branchings are completely independent, it is natural to postulate that 
\be
\Lambda_\mathbf{s} = \prod_{i=1}^N \Lambda_{s_i},
\ee
that is, that the measure of the set of worlds with outcomes $\mathbf{s}$ is the \emph{product} of the single-trial measures, and this is enough to show that in most worlds the relative frequencies will be right. Defining the proportion of worlds with outcome sequence $\mathbf{s}$ to be
\be 
\lambda_\mathbf{s} := \frac{\Lambda_\mathbf{s}}{\sum_\mathbf{s'} \Lambda_\mathbf{s'}},
\ee
it follows that the proportion of worlds with relative frequency $k/N$ of outcome $1$ is
\be \lambda(k,N) = \binom{N}{k} \lambda_1^k (1-\lambda_1)^{N-k}, \ee
which is again formally identical to the binomial distribution, and thus the law of large numbers for the proportion of worlds follows.

This implies that in any many-worlds theory with such a product measure it is eminently reasonable to count relative frequencies, and use them to guess which is the proportion of worlds with a given outcome: in most worlds one will be right.

This works perfectly well in Many-Worlds, if we take the measure of worlds to be the one suggested by the Born rule, or, more generally, the $p$-Born rule. To see that, consider again the $N$ repetitions of the $n$-outcome measurement, which in Many-Worlds is represented by the transformation
\be
\bigotimes_{i=1}^N \sum_{j=1}^n \alpha_j\ket{j}\ket{M_?^i} \mapsto \ket{w} = \bigotimes_{i=1}^N \sum_{j=1}^n \alpha_j\ket{j}\ket{M_j^i},
\ee
where the final state $\ket{w}$ can be written as the superposition of the $n^N$ sets of worlds
\be 
\ket{w} = \sum_{\mathbf{s}} \ket{w_\mathbf{s}} = \sum_{\mathbf{s}}  \prod_{i=1}^N\alpha_{s_i}\bigotimes_{j=1}^N \ket{s_j}\ket{M_{s_j}^j}.
\ee
If we then postulate the measure of the set of worlds $\ket{w_\mathbf{s}}$ to be
\be\label{eq:measurep}
\Lambda_\mathbf{s}  := \norm{\ket{w_\mathbf{s}}}_p^p,
\ee
then it does decompose as a product of the single-trial measures, as required. Defining again the proportion of worlds with outcomes $\mathbf{s}$ to be
\be 
\lambda_\mathbf{s} := \frac{\Lambda_\mathbf{s}}{\sum_\mathbf{s'} \Lambda_\mathbf{s'}},
\ee
we see that the proportion of worlds with relative frequency $f_N = k/N$ is again given by the binomial distribution, from which the law of large numbers for the proportion of worlds again follows.

In summary, if the correct way to measure the proportion of worlds is via the $p$-norm, then in most worlds we will observe frequencies conforming to the $p$-Born rule. Conversely, if we observe relative frequencies conforming to the $2$-Born rule, we should bet that the correct way to measure the proportion of worlds is via the $2$-norm.

\section{Multiplying or splitting}

\begin{figure}[ht]
\centering
\includegraphics{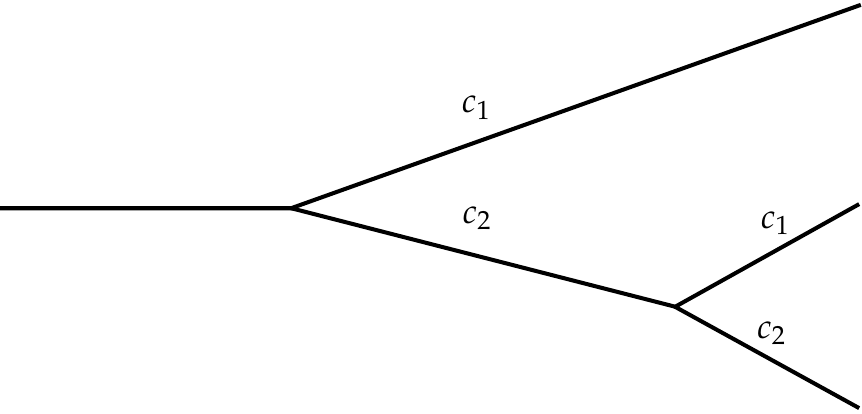}
\caption{An agent performs a measurement described by the coefficients $c_1$ and $c_2$. In the worlds where outcome $2$ is obtained, the agent performs the same measurement again.}\label{fig:onceortwice}
\end{figure}

The analogy between Kent's universe and Many-Worlds is not perfect, however, and it breaks down in the Once-or-Twice scenario shown in Fig.~\ref{fig:onceortwice} (discussed by Wallace in Ref.~\cite{wallace12} and Sebens and Carroll in Ref.~\cite{sebens14}). In it, an agent performs a measurement described by the coefficients $c_1$ and $c_2$, and in the worlds where outcome $2$ is obtained they perform the same measurement again. What are the subjective and objective probabilities of the sequences of outcomes $(1)$, $(2,1)$, and $(2,2)$?

First we shall consider the objective probabilities. In Kent's universe it seems clear what they are: there are in total $c_1+c_1c_2+c_2^2$ worlds, of which $c_1$ have outcome $(1)$, $c_2c_1$ have outcomes $(2,1)$, and $c_2^2$ have outcomes $(2,2)$, so the proportions are
\begin{subequations}\label{eq:multiplyingmeasure}
\begin{gather}
 \sharp_{1} = \frac{c_1}{c_1+c_1c_2+c_2^2}, \\
 \sharp_{21} = \frac{c_2c_1}{c_1+c_1c_2+c_2^2}, \\
 \sharp_{22}  = \frac{c_2^2}{c_1+c_1c_2+c_2^2}.
\end{gather}
\end{subequations}
If we measure frequencies, which are after all what objective probabilities should be connected to, then in most worlds they will be close to these proportions.

In Many-Worlds (and more generally in the hypothetical $p$-theories), however, Once-or-Twice is described by the transformation
\begin{multline}
(c_1\ket{1}+c_2\ket{2})\ket{M_?}(c_1\ket{1}+c_2\ket{2})\ket{M_?} \mapsto \\  c_1\ket{1}\ket{M_1}(c_1\ket{1}+c_2\ket{2})\ket{M_?} + c_2c_1\ket{2}\ket{M_2}\ket{1}\ket{M_1} + c_2^2\ket{2}\ket{M_2}\ket{2}\ket{M_2},
\end{multline}
and the measure of worlds is given by equation \eqref{eq:measurep}, which tells us that
\begin{subequations}\label{eq:measure}
\begin{gather}
 \lambda_{1}  = \frac{|c_1|^p}{|c_1|^p+|c_2|^p},\\
 \lambda_{21} = \frac{|c_2|^p|c_1|^p}{(|c_1|^p+|c_2|^p)^2}, \\
 \lambda_{22} = \frac{|c_2|^{2p}}{(|c_1|^p+|c_2|^p)^2}, 
\end{gather}
\end{subequations}
a fundamentally different result. If $c_1=c_2$, for example, $\sharp_1 = 1/3$, but $\lambda_1 = 1/2$.

What is the source of this disconnect? Our proposal is that while in Kent's universe the number of worlds is quite literally being multiplied upon branching, this is clearly not the case in Many-Worlds, as branching increases neither number of particles nor energy. In fact, the total measure of worlds, as given by equation \eqref{eq:measurep}, is conserved. A closer analogue to Many-Worlds would be a reversed version of Kent's universe where branching conserves the number of worlds. There the computer simulation starts with a large number of worlds, and when a measurement is made each outcome is imprinted in a subset of the worlds, with the relative sizes of these subsets determined by the coefficients. In the Once-or-Twice case, it suffices to start with $(c_1+c_2)^2$ worlds. After the first branching $(c_1+c_2)c_1$ are imprinted with outcome $1$, and $(c_1+c_2)c_2$ with outcome $2$. Out of the $(c_1+c_2)c_2$ with outcome $2$ a further $c_2c_1$ receive a $1$, and $c_2^2$ receive another $2$. The proportions of worlds are given by
\begin{subequations}\label{eq:splittingmeasure}
\begin{gather}
 \sharp'_{1} = \frac{c_1}{c_1+c_2}, \\
 \sharp'_{21} = \frac{c_2c_1}{(c_1+c_2)^2}, \\
 \sharp'_{22} = \frac{c_2^2}{(c_1+c_2)^2}, 
\end{gather}
\end{subequations}
now matching Many-Worlds\footnote{It should be clear that neither the regular nor the reverse Kent's universe are realistic analogues of Many-Worlds: both of them require an exponential number of worlds, either in the end or in the beginning of the simulation.}.

Now turning to the subjective probabilities, we face a difficulty, because the generalised Deutsch-Wallace theorem proved here deals only with simple measurements, not sequences of measurements. The original versions by Deutsch and Wallace do attribute probabilities to sequences of measurements, though, through an axiom that Deutsch called substitutibility and Wallace called diachronic consistency. It says that the value of a game does not change if one of its rewards is replaced by a game of equal value, so if an agent attributes value $V(G)$ to some game $G$, they really must be indifferent between receiving reward $V(G)$ or playing $G$, even when this reward would be given as part of another game. More formally, it is
\begin{itemize}
 \item \textbf{Substitution} \textit{Let $G = (\mathbf{c},\mathbf{r})$ be a game, and $G' = (\mathbf{c}',\mathbf{r}')$ another game such that $V(G') = r_1$. Then the sequential game where an agent plays $G$ but instead of receiving reward $r_1$ they play $G'$ has value equal to $V(G)$.}
\end{itemize}
With it, we can prove that the subjective probabilities in Many-Worlds match the proportion of worlds from equations \eqref{eq:measure}, as expected, but \textbf{Substitution} also implies that in both the normal and the reverse Kent's universe the subjective probabilities match those from equations \eqref{eq:splittingmeasure}. Here the Principal Principle raises a red flag, as in the regular Kent's universe they should match equations \eqref{eq:multiplyingmeasure} instead.

What went wrong? Well, \textbf{Substitution} implies that, in the case where $c_1=c_2=1$, one should be indifferent between playing a game where one future self gets reward $r_1$ and another reward $r_2$, and another game where one future self gets reward $r_1$ but two get reward $r_2$. If all these future selves are equal, as is the case in Kent's universe, one should definitely not be indifferent! One should prefer the second game, and assign probability $1/3$ to each outcome\footnote{Sebens and Carroll claim, however, that an agent that assigns these probabilities can be Dutch-booked \cite{sebens14}: After the first measurement, but before the second, an agent that is ignorant of the outcome will assign probability $1/2$ to being in each world, and will therefore accept a bet that pays $3$€ in the world with outcome $1$ and $-3$€ in the world with outcome $2$. After the second measurement, the agent will assign probability $1/3$ to being in each world, and will therefore accept a bet that pays $-4$€ in the world with a single outcome $1$ and $2$€ in the worlds with outcomes $(2,1)$ and $(2,2)$. If the agent accepts both bets, then in all worlds they will lose $1$€.

The problem with this argument is that the agent would not have accepted the first bet if they knew that they would be multiplied in the world with outcome $2$: the assigment of probability $1/2$ was a mistake, that the agent corrected after learning of the second multiplication. Given that they accepted the incorrect bet, however, the agent knows that accepting the second, fair, bet will cause them to be Dutch-booked, so they would reject it. Note, furthermore, that making two bets about the same situation using two different probabilities generically leads to Dutch-booking, even in single-world theories. Consider an agent that believes a coin to be fair, and therefore accepts a bet that pays $3$€ if heads, and $-3$€ if tails. If the agent changes their mind, and decides instead that the coin has probability $2/3$ of coming up tails, and accepts a bet that pays $-4$€ if heads and $2$€ if tails, then the agent will lose $1$€ independently of the result of the coin flip.

This raises the question of how a realistic agent in Kent's universe could ever assign probabilities, since they would depend on the whole tree of future branchings. The agent could simply recognize that it is not possible to determine the objective probabilities, and use only the part of the branching tree that they can foresee to calculate their subjective probabilities. Analogously, a grandparent could decide to divide their inheritance among their children weighted by how the number of grandchildren they begat, but refuse to speculate about how many children each grandchild will have.}. This is also the answer one obtains from Elga's indifference principle \cite{elga04}.

We should, therefore, reject \textbf{Substitution} as a rationality principle valid for general many-world theories, as it seems valid only for those where the total measure of worlds is conserved upon branching. As a substitute, in both the regular and the reverse Kent's universe we can unproblematically follow the objective probabilities, and say that a sequential game is equivalent to a simple game where the same rewards are given in the same number of worlds. In Many-Worlds the objective probabilities are not well-established, so instead we argue that different ways of doing a measurement are equivalent, as done in section \ref{sec:manyworlds}. Specifically, we say that measuring a state $\ket{\psi}$ and then a state $\ket{\varphi}$ in the worlds with outcome $2$ is equivalent to doing a joint measurement on $\ket{\psi}\ket{\varphi}$ and coarse-graining the outcomes of the measurement on $\ket{\varphi}$ in the worlds where the outcome of the measurement on $\ket{\psi}$ was different than $2$. This allows us to derive \textbf{Substitution} as a theorem in Many-Worlds and in the reverse Kent's universe, but not in the regular Kent's universe. This is done in detail in Appendix \ref{sec:reduction}.

In all many-world theories, these arguments imply that the agents should be indifferent to the number of branchings, caring only about which coefficients are associated to which rewards, thus satisfying the Principal Principle.

\section{Conclusion}

We investigated how subjective and objective probabilities work in many-world theories.

With respect to subjective probabilities, we generalised the Deutsch-Wallace theorem, and showed how it can be broken down in three parts: the first is a decision-theoretical core common to both single-world and many-world theories, that says that rational agents should use probabilities to make their decisions. The second part is a symmetry argument very natural for many-world theories, but less so for single-world ones, that implies that rational agents should assign uniform probabilities to physically symmetric games. The last part is a fine-graining argument used to reduce arbitrary games to symmetric ones. This reduction takes different forms depending on the precise physics of the many-worlds theory in question: in Kent's universe it implies that the probabilities come from the $1$-norm, while in Many-Worlds it implies that the probabilities come from the $2$-norm. While in Kent's universe the fine-graining is motivated explicitly from many-worlds considerations, the Many-Worlds version depends only on operational assumptions, so if one can stomach the symmetry argument, this derivation of the Born rule can also be considered valid for single-world versions of quantum mechanics.

With respect to objective probabilities, we have argued that they should be identified with the proportion of worlds in which an event happens, and showed that for any measure of worlds that has a product form an analogue of the law of large numbers can be proven: in most worlds the relative frequency will be close to the proportion of worlds. In Kent's universe, the motivating example, this argument can tell us what the objective probabilities are, as it is obvious how to measure the proportion of worlds. Furthermore, one can use the Principal Principle as a consistent check, and note that they match the subjective probabilities from the generalised Deutsch-Wallace theorem.

In Many-Worlds, however, it is not clear how to measure the proportion of worlds, so this argument cannot be used to derive the objective probabilities. It could be used in the other direction, though: if one takes both the Deutsch-Wallace theorem and the Principal Principle as true, then one could show that worlds should be measured via the $2$-norm. Alternatively, the argument can simply be used to say that if you want to find out what the proportion of worlds is, gathering relative frequencies is a good idea, as in most worlds you'll get the right answer.

We have also shown that the analogy between Many-Worlds and Kent's universe breaks down in the Once-or-Twice scenario, and used this breakdown to argue that a key rationality principle used in the original version of the Deutsch-Wallace theorem -- that when playing a sequential game agents should be indifferent between receiving a reward or playing a subgame with the same value -- is not valid in general many-world theories, and thus one should rely on other arguments to calculate the value of sequential games.

\section*{Acknowledgments}

We thank Koenraad Audenaert, Časlav Brukner, Eric Cavalcanti, Fabio Costa, Daniel Süß, David Gross, Markus Heinrich, Philipp Höhn, Felipe M. Mora, Jacques Pienaar, Simon Saunders, and David Wallace for useful discussions. This work has been supported by the Excellence Initiative of the German Federal and State Governments (Grant ZUK 81).

\bibliographystyle{linksen}
\bibliography{biblio}

\providecommand{\href}[2]{#2}\begingroup\raggedright\begin{thebibliography}{10}

\bibitem{everett57}
H.~Everett, ````Relative State'' Formulation of Quantum Mechanics,''
  \href{http://dx.doi.org/10.1103/RevModPhys.29.454}{{\em Rev. Mod. Phys.}
  {\bfseries 29}, 454--462 (1957)}.

\bibitem{finkelstein63}
D.~Finkelstein, ``The Logic of Quantum Physics,''
  \href{http://dx.doi.org/10.1111/j.2164-0947.1963.tb01483.x}{{\em Transactions
  of the New York Academy of Sciences} {\bfseries 25}, 621--637 (1963)}.

\bibitem{hartle68}
J.~B. Hartle, ``Quantum Mechanics of Individual Systems,''
  \href{http://dx.doi.org/10.1119/1.1975096}{{\em Am. J.~Phys.} {\bfseries 36},
  704 (1968)}.

\bibitem{dewitt70}
B.~S. DeWitt, ``Quantum mechanics and reality,''
  \href{http://dx.doi.org/10.1063/1.3022331}{{\em Phys. Today} {\bfseries 23},
  30 (1970)}.

\bibitem{graham73}
N.~Graham, ``The measurement of relative frequency,'' in {\em The many-worlds
  interpretation of quantum mechanics}, B.~DeWitt and N.~Graham, eds.
\newblock Princeton University Press, 1973.

\bibitem{farhi89}
E.~Farhi, J.~Goldstone, and S.~Gutmann, ``How probability arises in quantum
  mechanics,'' \href{http://dx.doi.org/10.1016/0003-4916(89)90141-3}{{\em Ann.
  Phys.} {\bfseries 192}, 368 -- 382 (1989)}.

\bibitem{squires90}
E.~J. Squires, ``On an alleged ``proof'' of the quantum probability law,''
  \href{http://dx.doi.org/10.1016/0375-9601(90)90192-Q}{{\em Phys. Lett.~A}
  {\bfseries 145}, 67 -- 68 (1990)}.

\bibitem{caves04b}
C.~M. {Caves} and R.~{Schack}, ``{Properties of the frequency operator do not
  imply the quantum probability postulate},''
  \href{http://dx.doi.org/10.1016/j.aop.2004.09.009}{{\em Ann. Phys.}
  {\bfseries 315}, 123--146 (2005)},
  \href{http://arxiv.org/abs/quant-ph/0409144}{{\ttfamily
  arXiv:quant-ph/0409144}}.

\bibitem{deutsch99}
D.~{Deutsch}, ``{Quantum theory of probability and decisions},''
  \href{http://dx.doi.org/10.1098/rspa.1999.0443}{{\em Proc. R. Soc. Lond. A}
  {\bfseries 455}, 3129 (1999)},
  \href{http://arxiv.org/abs/quant-ph/9906015}{{\ttfamily
  arXiv:quant-ph/9906015}}.

\bibitem{wallace03}
D.~{Wallace}, ``{Everettian Rationality: defending Deutsch's approach to
  probability in the Everett interpretation},''
  \href{http://dx.doi.org/10.1016/S1355-2198(03)00036-4}{{\em Stud. Hist. Phil.
  Sci.~B} {\bfseries 34}, 415 -- 439 (2003)},
  \href{http://arxiv.org/abs/quant-ph/0303050}{{\ttfamily
  arXiv:quant-ph/0303050}}.

\bibitem{wallace07}
D.~Wallace, ``Quantum probability from subjective likelihood: Improving on
  Deutsch's proof of the probability rule,''
  \href{http://dx.doi.org/10.1016/j.shpsb.2006.04.008}{{\em Stud. Hist. Phil.
  Sci.~B} {\bfseries 38}, 311 -- 332 (2007)},
  \href{http://arxiv.org/abs/quant-ph/0312157}{{\ttfamily
  arXiv:quant-ph/0312157}}.

\bibitem{wallace09}
D.~{Wallace}, ``{A formal proof of the Born rule from decision-theoretic
  assumptions},'' in {\em Many Worlds? Everett, Quantum Theory \& Reality},
  S.~Saunders, J.~Barrett, A.~Kent, and D.~Wallace, eds.
\newblock Oxford University Press, 2010.
\newblock \href{http://arxiv.org/abs/0906.2718}{{\ttfamily arXiv:0906.2718
  [quant-ph]}}.

\bibitem{zurek03}
W.~H. {Zurek}, ``{Environment-Assisted Invariance, Entanglement, and
  Probabilities in Quantum Physics},''
  \href{http://dx.doi.org/10.1103/PhysRevLett.90.120404}{{\em Phys. Rev. Lett.}
  {\bfseries 90}, 120404 (2003)},
  \href{http://arxiv.org/abs/quant-ph/0211037}{{\ttfamily
  arXiv:quant-ph/0211037}}.

\bibitem{vaidman11}
L.~Vaidman, ``Probability in the Many-Worlds Interpretation of Quantum
  Mechanics,'' in {\em The Probable and the Improbable: Understanding
  Probability in Physics, Essays in Memory of Itamar Pitowsky}, Y.~Ben-Menahem
  and M.~Hemmo, eds.
\newblock Springer, 2011.
\newblock \href{http://philsci-archive.pitt.edu/8558}{{\ttfamily
  PhilSci:8558}}.

\bibitem{carroll14}
S.~M. {Carroll} and C.~T. {Sebens},
  \href{http://dx.doi.org/10.1007/978-88-470-5217-8_10}{``Many Worlds, the Born
  Rule, and Self-Locating Uncertainty,''} in {\em Quantum Theory: A Two-Time
  Success Story}, D.~C. {Struppa} and J.~M. {Tollaksen}, eds., p.~157.
\newblock Springer Verlag, 2014.
\newblock \href{http://arxiv.org/abs/1405.7907}{{\ttfamily arXiv:1405.7907
  [gr-qc]}}.

\bibitem{kent09}
A.~{Kent}, ``One world versus many: the inadequacy of Everettian accounts of
  evolution, probability, and scientific confirmation,'' in {\em Many Worlds?
  Everett, Quantum Theory \& Reality}, S.~Saunders, J.~Barrett, A.~Kent, and
  D.~Wallace, eds.
\newblock Oxford University Press, 2010.
\newblock \href{http://arxiv.org/abs/0905.0624}{{\ttfamily arXiv:0905.0624
  [quant-ph]}}.

\bibitem{gill03}
R.~D. {Gill}, ``{On an Argument of David Deutsch},'' in {\em Quantum
  Probability and Infinite Dimensional Analysis: from Foundations to
  Applications}, M.~Schürmann and U.~Franz, eds., pp.~277--292.
\newblock World Scientific, 2005.
\newblock \href{http://arxiv.org/abs/quant-ph/0307188}{{\ttfamily
  arXiv:quant-ph/0307188}}.

\bibitem{baker06}
D.~J. Baker, ``Measurement outcomes and probability in Everettian quantum
  mechanics,'' \href{http://dx.doi.org/10.1016/j.shpsb.2006.05.003}{{\em Stud.
  Hist. Phil. Sci.~B} {\bfseries 38}, 153 -- 169 (2007)},
  \href{http://philsci-archive.pitt.edu/2717}{{\ttfamily PhilSci:2717}}.

\bibitem{albert10}
D.~Albert, ``Probability in the Everett Picture,'' in {\em Many Worlds?
  Everett, Quantum Theory \& Reality}, S.~Saunders, J.~Barrett, A.~Kent, and
  D.~Wallace, eds.
\newblock Oxford University Press, 2010.

\bibitem{dawid13}
R.~Dawid and K.~Th{\'e}bault, ``Many worlds: decoherent or incoherent?,''
  \href{http://dx.doi.org/10.1007/s11229-014-0650-8}{{\em Synthese} {\bfseries
  192}, 1559--1580 (2015)},
  \href{http://philsci-archive.pitt.edu/9542}{{\ttfamily PhilSci:9542}}.

\bibitem{mandolesi15}
A.~L.~G. {Mandolesi}, ``Analysis of Wallace's Proof of the Born Rule in
  Everettian Quantum Mechanics: Formal Aspects,''
  \href{http://dx.doi.org/10.1007/s10701-018-0179-7}{{\em Found. Phys.}
  {\bfseries 48}, 751--782 (2018)},
  \href{http://arxiv.org/abs/1504.05259}{{\ttfamily arXiv:1504.05259
  [quant-ph]}}.

\bibitem{mandolesi18}
A.~L.~G. {Mandolesi}, ``Analysis of Wallace's Proof of the Born Rule in
  Everettian Quantum Mechanics II: Concepts and Axioms,''
  \href{http://dx.doi.org/10.1007/s10701-018-0226-4}{{\em Found. Phys.}
  {\bfseries 49}, 24--52 (2019)},
  \href{http://arxiv.org/abs/1803.08762}{{\ttfamily arXiv:1803.08762
  [quant-ph]}}.

\bibitem{wallace12}
D.~Wallace, {\em The Emergent Multiverse: Quantum Theory According to the
  Everett Interpretation}.
\newblock Oxford University Press, 2012.

\bibitem{joyce08}
J.~Joyce, {\em The Foundations of Causal Decision Theory}.
\newblock Cambridge University Press, 2008.

\bibitem{savage54}
L.~Savage, {\em The Foundations of Statistics}.
\newblock Dover Publications, 1954 (1972).

\bibitem{vanfraassen89}
B.~{van Fraassen},
  \href{http://dx.doi.org/10.1093/0198248601.003.0012}{``Indifference: The
  Symmetries of Probability,''} in {\em Laws and Symmetry}.
\newblock Oxford University Press, 1989.

\bibitem{saunders93}
S.~Saunders, ``Decoherence, relative states, and evolutionary adaptation,''
  \href{http://dx.doi.org/10.1007/BF00732365}{{\em Found. Phys.} {\bfseries
  23}, 1553--1585 (1993)}.

\bibitem{wallace01}
D.~{Wallace}, ``Everett and Structure,''
  \href{http://dx.doi.org/10.1016/S1355-2198(02)00085-0}{{\em Stud. Hist. Phil.
  Sci.~B} {\bfseries 34}, 87--105 (2003)},
  \href{http://arxiv.org/abs/quant-ph/0107144}{{\ttfamily
  arXiv:quant-ph/0107144}}.

\bibitem{zurek03b}
W.~H. {Zurek}, ``{Decoherence, einselection, and the quantum origins of the
  classical},'' \href{http://dx.doi.org/10.1103/RevModPhys.75.715}{{\em Rev.
  Mod. Phys.} {\bfseries 75}, 715--775 (2003)},
  \href{http://arxiv.org/abs/quant-ph/0105127}{{\ttfamily
  arXiv:quant-ph/0105127}}.

\bibitem{hartle08}
J.~B. {Hartle}, ``The Quasiclassical Realms of This Quantum Universe,''
  \href{http://dx.doi.org/10.1007/s10701-010-9460-0}{{\em Found. Phys.}
  {\bfseries 41}, 982--1006 (2011)},
  \href{http://arxiv.org/abs/0806.3776}{{\ttfamily arXiv:0806.3776
  [quant-ph]}}.

\bibitem{price10}
H.~Price, ``Decisions, decisions, decisions: can Savage salvage Everettian
  probability?,'' in {\em Many Worlds? Everett, Quantum Theory \& Reality},
  S.~Saunders, J.~Barrett, A.~Kent, and D.~Wallace, eds.
\newblock Oxford University Press, 2010.
\newblock \href{http://philsci-archive.pitt.edu/3886}{{\ttfamily
  PhilSci:3886}}.

\bibitem{gerlach22}
W.~Gerlach and O.~Stern, ``Der experimentelle Nachweis der Richtungsquantelung
  im Magnetfeld,'' \href{http://dx.doi.org/10.1007/BF01326983}{{\em Zeitschrift
  f{\"u}r Physik} {\bfseries 9}, 349--352 (1922)}.

\bibitem{phipps32}
T.~E. Phipps and O.~Stern, ``{\"U}ber die Einstellung der
  Richtungsquantelung,'' \href{http://dx.doi.org/10.1007/BF01351212}{{\em
  Zeitschrift f{\"u}r Physik} {\bfseries 73}, 185--191 (1932)}.

\bibitem{frisch33}
R.~Frisch and E.~Segrè, ``{\"U}ber die Einstellung der Richtungsquantelung.
  II,'' \href{http://dx.doi.org/10.1007/BF01335699}{{\em Zeitschrift f{\"u}r
  Physik} {\bfseries 80}, 610--616 (1933)}.

\bibitem{sakurai93}
J.-J. Sakurai, {\em Modern Quantum Mechanics}.
\newblock Addison Wesley, 1993.

\bibitem{feynman77}
R.~Feynman, R.~Leighton, and M.~Sands, {\em The Feynman Lectures on Physics,
  Volume 3}.
\newblock Addison-Wesley, 1977.

\bibitem{landau65}
L.~Landau and E.~Lifshitz, {\em Quantum Mechanics: Non-Relativistic Theory}.
\newblock Pergamon, 1965.

\bibitem{vonneumann32}
J.~von Neumann, {\em Mathematische Grundlagen der Quantenmechanik}.
\newblock Springer, 1932.

\bibitem{chuang00}
M.~Nielsen and I.~Chuang, {\em Quantum Computation and Quantum Information}.
\newblock Cambridge University Press, 2000.

\bibitem{peres06}
A.~Peres, {\em Quantum Theory: Concepts and Methods}.
\newblock Springer Netherlands, 2006.

\bibitem{schroedinger26}
E.~Schrödinger, ``Quantisierung als Eigenwertproblem,''
  \href{http://dx.doi.org/10.1002/andp.19263840404}{{\em Annalen der Physik}
  {\bfseries 384}, 361--376 (1926)}.

\bibitem{born26}
M.~Born, ``Zur Quantenmechanik der Sto{\ss}vorg{\"a}nge,''
  \href{http://dx.doi.org/10.1007/BF01397477}{{\em Zeitschrift f{\"u}r Physik}
  {\bfseries 37}, 863--867 (1926)}.

\bibitem{aaronson04}
S.~{Aaronson}, ``Is Quantum Mechanics An Island In Theoryspace?,''
  \href{http://arxiv.org/abs/quant-ph/0401062}{{\ttfamily
  arXiv:quant-ph/0401062}}.

\bibitem{lewis80}
D.~Lewis, \href{http://dx.doi.org/10.1093/0195036468.003.0004}{``A
  Subjectivist's Guide to Objective Chance,''} in {\em Studies in Inductive
  Logic and Probability, Vol II.}, R.~C. Jeffrey, ed.
\newblock University of California Press, 1980.

\bibitem{saunders10}
S.~{Saunders}, ``Chance in the Everett interpretation,'' in {\em Many Worlds?
  Everett, Quantum Theory \& Reality}, S.~Saunders, J.~Barrett, A.~Kent, and
  D.~Wallace, eds.
\newblock Oxford University Press, 2010.
\newblock \href{http://arxiv.org/abs/1609.04720}{{\ttfamily arXiv:1609.04720
  [quant-ph]}}.

\bibitem{sebens14}
C.~T. {Sebens} and S.~M. {Carroll}, ``Self-Locating Uncertainty and the Origin
  of Probability in Everettian Quantum Mechanics,''
  \href{http://dx.doi.org/10.1093/bjps/axw004}{{\em Br. J. Philos. Sci.}
  {\bfseries 69}, 25--74 (2018)},
  \href{http://arxiv.org/abs/1405.7577}{{\ttfamily arXiv:1405.7577
  [quant-ph]}}.

\bibitem{elga04}
A.~Elga, ``Defeating Dr. Evil with Self-Locating Belief,''
  \href{http://dx.doi.org/10.1111/j.1933-1592.2004.tb00400.x}{{\em Philosophy
  and Phenomenological Research} {\bfseries 69}, 383--396 (2004)}.

\bibitem{bohnenblust40}
F.~Bohnenblust, ``An axiomatic characterization of $L_p$-spaces,''
  \href{http://dx.doi.org/10.1215/S0012-7094-40-00648-2}{{\em Duke Math. J.}
  {\bfseries 6}, 627--640 (1940)}.

\end{thebibliography}\endgroup

\appendix

\section{Consistent fine-graining}\label{sec:consistency}

For which norms can the fine-graining argument work?

There are multiple ways to fine-grain a game with unequal coefficients into a symmetric game. For example, if $\mu = \big\|\big(\norm{\mathbf{1}^{(n)}},\norm{\mathbf{1}^{(m)}}\big)\big\|$, where $\mathbf{1}^{(n)}$ is the vector with $n$ ones, one can fine-grain the game
\be G = \begin{pmatrix} 1 & r_1 \\ \mu & r_2 \end{pmatrix} \ee
either by first taking it to
\be G' = \begin{pmatrix} 1 & r_1 \\ \norm{\mathbf{1}^{(n)}} & r_2 \\ \norm{\mathbf{1}^{(m)}} & r_2 \end{pmatrix}, \ee
and then applying two more fine-grainings to take it to
\be G'' = \begin{pmatrix} 1 & r_1 \\ \mathbf{1}^{(n)} & \mathbf{r}_2^{(n)} \\ \mathbf{1}^{(m)} & \mathbf{r}_2^{(m)} \end{pmatrix}, \ee
or by taking $G$ directly to $G''$, which will be possible only if $\mu = \norm{\mathbf{1}^{(n+m)}}$. We want all possible ways of fine-graining a game to give same result, so we demand the norm to be such that for all vectors $\mathbf{v}$ and $\mathbf{w}$ with disjoint support
\be\norm{\mathbf{v} + \mathbf{w}} = \big\|\big(\norm{\mathbf{v}},\norm{\mathbf{w}}\big)\big\|.\ee
We also demand the norm to be permutation-invariant, as it seems unphysical to attribute meaning to the labelling of the vectors, and that $\norm{(1,1)} \neq 1$, because otherwise $\norm{\mathbf{1}^{(n)}}=1$ for all $n$, and it is therefore impossible to fine-grain any non-trivial game.

These conditions are enough to show that these norms must be equivalent to $p$-norms when restricted to vectors of rational numbers, as can be seen by adapting an argument by Bohnenblust \cite{bohnenblust40}. We have then
\begin{theorem}
Let $\|\cdot\|:\mathbb C^N \to \mathbb R$ be a permutation-invariant norm for $N\ge 3$ such that $\norm{(1,1)} \neq 1$ and
\be\norm{\mathbf{v} + \mathbf{w}} = \big\|\big(\norm{\mathbf{v}},\norm{\mathbf{w}}\big)\big\|\ee
for all vectors $\mathbf{v},\mathbf{w}$ with disjoint support. Then for any vector $\mathbf{c}$ such that the absolute value of its components is rational, \be\norm{\mathbf{c}} = \Big(\sum_i |c_i|^p\Big)^{\frac1p},\ee
for some real number $p \ge 1$.
\end{theorem}
\begin{proof}
Let $f$ be such that $f(1):=\norm{1}$ and $f(n+1) := \norm{(1,f(n))}$. First note that $f(1) = 1$, as $\norm{1} = \norm{(1,0)} = \norm{(\norm{(1,0)},\norm{(0,0)})} = \norm{(1,0)}^2$. 

We need to show that $f(n)$ is monotonous, and that $f(n^k) = f^k(n)$. For the former, consider the identity $2(f(n),0) = (f(n),1) + (f(n),-1)$ and take the norm of both sides. By the triangle inequality
\be 2f(n) \le \norm{(f(n),1)} + \norm{(f(n),-1)} = 2f(n+1).\ee
For the latter, first we show that $f(n+m) = \norm{(f(m),f(n)}$. Assume that it holds for some $m$. Then 
\be f(n+1+m) = \norm{(f(m),f(n+1))} = \norm{(f(m),1,f(n))} = \norm{(f(m+1),f(n))}.\ee
Since it holds for $m=1$, by induction it holds for all $m$. Now assume that $f(nm) = f(n)f(m)$ holds for some $m$. Then 
\be 
f(n(m+1)) = f(nm+n) = \norm{(f(nm),f(n))} = f(n)\norm{(f(m),1)} = f(n)f(m+1),
\ee
and therefore by induction this is true for all $m$, as it obviously holds for $m=1$. This implies that $f(n^k) = f^k(n)$.

Now let $m,n\ge 2$ be some fixed integers, and $h$ the integer such that for any positive integer $k$
\be
m^h \le n^k < m^{h+1}.
\ee
Applying $f(\cdot)$ to these numbers, it follows that
\be
h\log f(m) \le k \log f(n) \le (h+1) \log f(m),
\ee
and using the fact that $h \le k \frac{\log n}{\log m}$ and $h > k \frac{\log n}{\log m}-1$ we have that
\be
\frac{\log f(m)}{\log m}-\frac{\log f(m)}{k\log n} < \frac{\log f(n)}{\log n} \le \frac{\log f(m)}{\log m}+\frac{\log f(m)}{k\log n}.
\ee
Taking the limit of $k$ going to infinity lets us conclude that 
\be
\frac{\log f(m)}{\log m} = \frac{\log f(n)}{\log n},
\ee
which means that this fraction is a constant independent of $n$ (and different than $0$ as $f(2)>1)$. Calling this constant $1/p$, we conclude that
$$f(n) = n^\frac1p.$$
Now for any rational number $m/n$ we have that
$$\|(1,m/n)\| = \frac1n\|(n,m)\| = \frac1n\|(f(n^p),f(m^p)\| = \frac1nf(n^p+m^p) = (1+(m/n)^p)^\frac1p,$$
so by homogeneity $\norm{(a,b)} = (|a|^p + |b|^p)^{\frac1p}$ for any rationals $|a|$ and $|b|$, and by induction for any vector $\mathbf{c}$ such that the absolute values of the components are rational numbers 
\be\norm{\mathbf{c}} = \Big(\sum_i |c_i|^p\Big)^{\frac1p}.\ee
\end{proof}
If one furthermore assumes some regularity condition, then the result is valid for any complex vector. 

\section{Reducing sequential games to simple games}\label{sec:reduction}

Consider the sequential game
\be
G = \begin{pmatrix}
 c_1 & r_1 \\
 c_2 & d_1 & s_1 \\
& d_2 & s_2 
 \end{pmatrix},
\ee
 where in the worlds with outcome $1$ reward $r_1$ is given, but in the worlds with outcome $2$ the subgame $H = \begin{pmatrix} d_1 & s_1 \\ d_2 & s_2 \end{pmatrix}$ is played. We want to reduce it to a simple game in both Kent's universes and Many-Worlds.
 
 In the regular Kent's universe $c_1$ worlds are created with reward $r_1$, $c_2d_1$ worlds are created with reward $s_1$, and $c_2d_2$ worlds are created with reward $s_2$, so it seems natural to postulate that $G$ is equivalent to the simple game
 \be 
 G' = \begin{pmatrix}
 c_1 & r_1 \\
 c_2d_1 & s_1 \\
 c_2d_2 & s_2 
 \end{pmatrix}.
 \ee
 In the reverse Kent's universe, there are $M(c_1+c_2)(d_1+d_2)$ worlds in the beginning. After the first branching, $Mc_1(d_1+d_2)$ worlds are imprinted with outcome $1$, and the remaining $Mc_2(d_1+d_2)$ worlds are split again, with $Mc_2d_1$ being imprinted with a further outcome $1$, and $Mc_2d_2$ with outcome $2$. In the end there are $Mc_1(d_1+d_2)$ with reward $r_1$, $Mc_2d_1$ worlds with reward $s_1$, and $Mc_2d_2$ worlds with reward $s_2$, so it seems natural to postulate that $G$ is equivalent to the simple game
 \be 
 G'' = \begin{pmatrix}
 c_1(d_1+d_2) & r_1 \\
 c_2d_1 & s_1 \\
 c_2d_2 & s_2 
 \end{pmatrix}.\ee
Note that in the reverse Kent's universe \textbf{Substitution} is satisfied: the value of the subgame $H$ is 
\be
V(H) = \frac{1}{d_1+d_2}(d_1s_1+d_2s_2),
\ee
and the value of $G$ there is
\begin{align}
V(G) &= \frac{1}{(c_1+c_2)(d_1+d_2)}\Big(c_1(d_1+d_2)r_1 + c_2(d_1s_1 + d_2s_2) \Big)  \\ &= \frac{1}{c_1+c_2}(c_1r_1+c_2V(H)),
\end{align}
which is equal to the value of $\begin{pmatrix} c_1 & r_1 \\ c_2 & V(H) \end{pmatrix}$, as required. We shall not prove the general case, as that is quite straightforward.

In Many-Worlds, the game $G$ is instantiated by making a measurement on the state $\ket{\psi} = c_1\ket{1} + c_2\ket{2},$ giving reward $r_1$ in the worlds with outcome $1$, and in the worlds with outcome $2$ doing a measurement on the state $\ket{\varphi} = d_1\ket{1} + d_2\ket{2},$ finally giving rewards $s_1$ and $s_2$ in the worlds with the second outcomes $1$ and $2$. These measurements take the initial state $\ket{\psi}\ket{M_?}\ket{\varphi}\ket{D_?}$ to the final state
\be
\ket{G} = c_1\ket{1}\ket{M_1}\ket{\varphi}\ket{D_?}\ket{r_1} + c_2d_1\ket{2}\ket{M_2}\ket{1}\ket{D_1}\ket{r_2} + c_2d_2\ket{2}\ket{M_2}\ket{2}\ket{D_2}\ket{r_3}.
\ee
An equivalent way to play this game is to make a joint measurement on the state
\be
\ket{\psi}\ket{\varphi} = c_1d_1\ket{1}\ket{1} + c_1d_2\ket{1}\ket{2} + c_2d_1\ket{2}\ket{1} + c_2d_2\ket{2}\ket{2},
\ee
but in the worlds where the measurement on $\ket{\psi}$ resulted in $1$ apply the label $?$ to both outcomes of the measurement on $\ket{\varphi}$, leading to the final state
\begin{multline}
 \ket{G'} = c_1d_1\ket{1}\ket{M_1}\ket{1}\ket{D_?'}\ket{r_1} + c_1d_2\ket{1}\ket{M_1}\ket{2}\ket{D_?''}\ket{r_1}
 \\+  c_2d_1\ket{2}\ket{M_2}\ket{1}\ket{D_1}\ket{s_1} + c_2d_2\ket{2}\ket{M_2}\ket{2}\ket{D_2}\ket{s_2}
\end{multline}
where $\ket{D_?'}$ and $\ket{D_?''}$ are measurements results physically distinct from $\ket{D_?}$, but with the same label.

If one does not, however, coarse-grain the results $(1,1)$ and $(1,2)$ together, then this measurement procedure can be regarded as playing the simple game
\be 
G''' =  \begin{pmatrix} 
c_1d_1 & r_1 \\
c_1d_2 & r_1 \\
c_2d_1 & s_1 \\
c_2d_2 & s_2 \end{pmatrix}
\ee
instead. We postulate therefore that a rational agent in Many-Worlds should regard $G$ and $G'''$ as equivalent, or more formally that:
\begin{itemize}
 \item \textbf{Reduction} \textit{The sequential game
 \be G = \begin{pmatrix}
 \alpha_1 & r_1 \\
 \vdots & \vdots \\
 \alpha_n & \beta_1 & s_1 \\
& \vdots & \vdots \\
& \beta_m & s_m 
 \end{pmatrix}\ee
where subgame $(\boldsymbol{\beta},\mathbf{s})$ is played in the worlds with outcome $n$, has the same value as the simple game 
\be G' = \begin{pmatrix}
 \alpha_1\boldsymbol{\beta} & r_1 \\
 \vdots & \vdots \\
 \alpha_{n-1}\boldsymbol{\beta} & r_{n-1} \\ \alpha_n\beta_1 & s_1 \\
 \vdots & \vdots \\
 \alpha_n\beta_{m-1} & s_{m-1} \\
 \alpha_n\beta_m & s_m 
 \end{pmatrix}.
\ee}
\end{itemize}
This \textbf{Reduction} postulate suffices to prove \textbf{Substitution} as a theorem, as the value of the subgame $(\boldsymbol{\beta},\mathbf{s})$ is
\be
V(\boldsymbol{\beta},\mathbf{s}) = \frac{1}{\norm{\boldsymbol{\beta}}^p_p}\sum_{j=1}^m|\beta_j|^p,
\ee
and the value of $G$ is
\begin{gather}
V(G) = \frac{1}{\norm{\boldsymbol{\alpha}}^p_p\norm{\boldsymbol{\beta}}^p_p}\Big(\norm{\boldsymbol{\beta}}^p_p\sum_{i=1}^{n-1}|\alpha_i|^pr_i  + |\alpha_n|^p\sum_{j=1}^m|\beta_j|^p s_j\Big) \\
     = \frac{1}{\norm{\boldsymbol{\alpha}}^p_p}\Big(\sum_{i=1}^{n-1}|\alpha_i|^pr_i  + |\alpha_n|^p V(\boldsymbol{\beta},\mathbf{s})\Big),
\end{gather}
as required.
  
\end{document}